\documentclass[a4paper,UKenglish,cleveref, autoref, thm-restate]{lipics-v2021}

\pdfoutput=1 \hideLIPIcs  

\graphicspath{{./figures/}}

\title{Parameterized Complexity of Vertex Splitting to Pathwidth at most 1}

\author{Jakob {Baumann}}{University of Passau, Germany}{baumannjak@fim.uni-passau.de}{https://orcid.org/0000-0002-2594-3828}{}

\author{Matthias {Pfretzschner}}{University of Passau, Germany}{pfretzschner@fim.uni-passau.de}{https://orcid.org/0000-0002-5378-1694}{}

\author{Ignaz {Rutter}}{University of Passau, Germany}{rutter@fim.uni-passau.de}{https://orcid.org/0000-0002-3794-4406}{}

\authorrunning{J. Baumann and M. Pfretzschner and I. Rutter} 

\Copyright{Jakob Baumann and Matthias Pfretzschner and Ignaz Rutter} 



\category{} 

\relatedversion{}

\nolinenumbers 

\EventEditors{John Q. Open and Joan R. Access}
\EventNoEds{2}
\EventLongTitle{42nd Conference on Very Important Topics (CVIT 2016)}
\EventShortTitle{CVIT 2016}
\EventAcronym{CVIT}
\EventYear{2016}
\EventDate{December 24--27, 2016}
\EventLocation{Little Whinging, United Kingdom}
\EventLogo{}
\SeriesVolume{42}
\ArticleNo{23}

\usepackage{xspace}
\usepackage{tabularx}
\usepackage{environ}
\usepackage{hyperref}
\usepackage{todonotes}
\usepackage{algorithm2e}
\usepackage{caption,subcaption}
\usepackage{mleftright}
\usepackage{mathtools}

\makeatletter
\newcommand{\problemtitle}[1]{\gdef\@problemtitle{#1}}
\newcommand{\probleminput}[1]{\gdef\@probleminput{#1}}
\newcommand{\problemquestion}[1]{\gdef\@problemquestion{#1}}
\NewEnviron{problem}{
  \problemtitle{}\probleminput{}\problemquestion{}
  \BODY
  \par\addvspace{.5\baselineskip}
  \noindent
  \begin{tabularx}{\textwidth}{@{\hspace{15pt}} l X c}
    \multicolumn{2}{@{\hspace{25pt}}l}{\textsc{\@problemtitle}} \\
    \textbf{Input:} & \@probleminput \\
    \textbf{Question:} & \@problemquestion
  \end{tabularx}
  \par\addvspace{.5\baselineskip}
}
\makeatother

\newcommand{\POVE}{\textsc{Pathwidth-One Vertex Explosion}\xspace}
\newcommand{\pove}{\textsc{POVE}\xspace}
\newcommand{\POVS}{\textsc{Pathwidth-One Vertex Splitting}\xspace}
\newcommand{\povs}{\textsc{POVS}\xspace}
\newcommand{\TOVS}{\textsc{Treewidth-One Vertex Splitting}\xspace}
\newcommand{\tovs}{\textsc{TOVS}\xspace}
\newcommand{\SPLITPI}{\textsc{$\Pi$ Vertex Splitting}\xspace}
\newcommand{\splitPi}{\textsc{$\Pi$-VS}\xspace}
\newcommand{\EXPLODEPI}{\textsc{$\Pi$ Vertex Explosion}\xspace}
\newcommand{\explodePi}{\textsc{$\Pi$-VE}\xspace}

\newcommand{\degs}{\ensuremath{\mathrm{deg^*}}}
\newcommand{\Sxpl}{\ensuremath{S_{\mathrm{explode}}}}
\newcommand{\Skeep}{\ensuremath{S_{\mathrm{keep}}}}

\newcommand{\PiExpl}{\ensuremath{\Pi_k^\times}}

\DeclareMathOperator{\adj}{adj}

\makeatletter
\newcommand{\mylabel}[1]{
  \let\oldlabel\@currentlabel
  \def\@currentlabel{(\oldlabel)}
  \label{#1}
}
\makeatother
\newtheorem{redrule}{Reduction Rule}
\newtheorem{branchrule}{Branching Rule}
\newtheorem{mytheorem}{Theorem}
\newtheorem{myproposition}{Proposition}

\newtheorem{mylemma}{Lemma}
\newtheorem{mycorollary}{Corollary}

\begin{document}
  
  \maketitle
  
  \begin{abstract}
     Motivated by the planarization of 2-layered straight-line drawings, we consider the problem of modifying a graph 
    such that the resulting graph has pathwidth at most 1. The problem \POVE (\pove) asks whether such a graph can be 
    obtained using at most $k$ vertex explosions, where a \emph{vertex explosion} replaces a vertex $v$ by $\deg(v)$ 
    degree-1 vertices, each incident to exactly one edge that was originally incident to~$v$. For \pove, we give an FPT 
    algorithm with running time $O(4^k \cdot m)$ and an $O(k^2)$ kernel, thereby improving over the $O(k^6)$-kernel by
    Ahmed et al.~\cite{Ahmed22} in a more general setting.
    Similarly, a \emph{vertex split} replaces a vertex~$v$ by two distinct vertices~$v_1$ and~$v_2$ and distributes the 
    edges originally incident to~$v$ arbitrarily to~$v_1$ and~$v_2$. Analogously to \pove, we define the problem 
    variant \POVS (\povs) that uses the split operation instead of vertex explosions. Here we obtain a linear kernel 
    and an algorithm with running time $O((6k+12)^k \cdot m)$.   This answers an open question by Ahmed et al.~\cite{Ahmed22}.
    
    Finally, we consider the problem \SPLITPI (\splitPi), which generalizes the problem \povs and asks whether a given graph can be turned into a graph of a specific graph class $\Pi$ using at most $k$ vertex splits. For graph classes $\Pi$ that can be tested in monadic second-order graph logic~(MSO$_2$), we show that the problem \splitPi can be expressed as an MSO$_2$ formula, resulting in an FPT algorithm for \splitPi parameterized by $k$ if $\Pi$ additionally has bounded treewidth.
    We obtain the same result for the problem variant using vertex explosions.
  \end{abstract}

  \section{Introduction}

  Crossings are one of the main aspects that negatively affect the readability of drawings~\cite{Purchase95}.  It is 
  therefore natural to try and modify a given graph in such a way that it can be drawn without crossings while 
  preserving as much of the information as possible.  We consider three different operations.

  A \emph{deletion operation} simply removes a vertex from the graph.  A \emph{vertex explosion} replaces a vertex~$v$ by~$\deg(v)$ degree-1 vertices, each incident to exactly one edge that was originally incident to~$v$.  Finally, a \emph{vertex split} replaces a vertex~$v$ by two distinct vertices~$v_1$ and~$v_2$ and distributes the edges originally incident to~$v$ arbitrarily to~$v_1$ and~$v_2$.

  Nöllenburg et al.~\cite{Nollenburg22} have recently studied the vertex splitting problem, which is known to be 
  NP-complete~\cite{Faria98}.  In particular, they gave a non-uniform FPT-algorithm for deciding whether a given graph 
  can be planarized with at most $k$ splits.
  
  We observe that, since degree-1 vertices can always be inserted into a planar drawing, the vertex explosion model and
  the vertex deletion model are equivalent for obtaining planar graphs. Note that this is not necessarily the case for
  other target graph classes (see, for example, Figure~\ref{fig:operations}). The problem of deleting vertices to
  obtain a planar graph is also known as \textsc{Vertex Planarization} and has been studied extensively in the
  literature. While the problem is NP-complete~\cite{Lewis80}, it follows from results of Robertson and
  Seymour~\cite{Robertson95} that the problem can be decided in cubic time for any fixed $k$. Subsequent algorithms
  gradually improved upon this result~\cite{Marx07, Kawarabayashi09}, culminating in an $O(2^{O(k \log k)}
  \cdot n)$-time algorithm introduced by Jansen et al.~\cite{Jansen14}.

  Ahmed et al.~\cite{Ahmed22} investigated the problem of splitting the vertices of a bipartite graph so 
  that it admits a 2-layered drawing without crossings.  They assume that the input graph is bipartite and only the 
  vertices of one of the two sets in the bipartition may be split.  Under this condition, they give an $O(k^6)$-kernel 
  for the vertex explosion model, which results in an $O(2^{O(k^6)}m)$-time algorithm.  They ask whether similar 
  results can be obtained in the vertex splitting model. Figure~\ref{fig:operations} illustrates the three operations 
  in the context of 2-layered drawings\footnote{In this context, minimizing the number of vertex explosions is equivalent to minimizing the number of vertices that are split, since it is always best to split a vertex as often as possible.}.
  
    \begin{figure}[t]
    \centering
    \begin{subfigure}[b]{0.3\textwidth}
      \centering
      \includegraphics[page=1,scale=1]{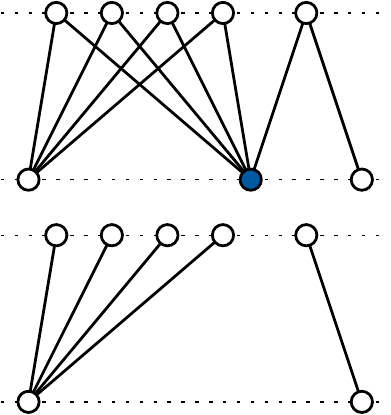}
      \caption{}
      \label{fig:twoLayerDeletion}
    \end{subfigure}
    \hfill
    \begin{subfigure}[b]{0.3\textwidth}
      \centering
      \includegraphics[page=2,scale=1]{OperationComparison}
      \caption{}
      \label{fig:twoLayerExplosion}
    \end{subfigure}
    \hfill
    \begin{subfigure}[b]{0.3\textwidth}
      \centering
      \includegraphics[page=3,scale=1]{OperationComparison}
      \caption{}
      \label{fig:twoLayerSplit}
    \end{subfigure}
    \caption{Given the shown bipartite graph, a crossing-free 2-layered drawing can be obtained using one vertex 
    deletion (a), two vertex explosions (b), or three vertex splits (c).}
    \label{fig:operations}
  \end{figure}

  We note that a graph admits a 2-layer drawing without crossings if and only if it has pathwidth at most~$1$, i.e., it is a disjoint union of caterpillars~\cite{Arnborg90,Eades86}.  Motivated by this, we more generally 
  consider the problem of turning a graph~$G=(V,E)$ into a graph of pathwidth at most~1 by the above operations.  In 
  order to model the restriction of Ahmed et al.~\cite{Ahmed22} that only one side of their bipartite input graph may 
  be split, we further assume that we are given a subset~$S \subseteq V$, to which we may apply modification operations 
  as part of the input. We define that the new vertices resulting from an operation are also included in $S$.

  More formally, we consider the following problems, all of which have been shown to be NP-hard 
  \cite{Ahmed23,Philip2010}.

  \begin{problem}
    \problemtitle{\POVE (\pove)}
    \probleminput{An undirected graph $G = (V, E)$, a set $S \subseteq V$, and a positive integer~$k$.}
    \problemquestion{Is there a set~$W \subseteq S$ with~$|W| \le k$ such that the graph resulting from exploding all vertices in~$W$ has pathwidth at most~1?}
  \end{problem}

   \begin{problem}
    \problemtitle{Pathwidth-One Vertex Splitting (POVS)}
    \probleminput{An undirected graph $G = (V, E)$, a set $S \subseteq V$, and a positive integer~$k$.}
    \problemquestion{Is there a sequence of at most $k$ splits on vertices in $S$ such that the resulting
      graph has pathwidth at most 1?}
  \end{problem}

  We refer to the analogous problem with the deletion operation as \textsc{Pathwidth-One Vertex Deletion} 
  (\textsc{POVD}).  Here an algorithm with running time $O(7^k \cdot n^{O(1)})$ is known~\cite{Philip2010}, which 
  was later improved to $O(4.65^k \cdot n^{O(1)})$~\cite{Cygan12}, and very recently to $O(3.888^k \cdot 
  n^{O(1)})$~\cite{Tsur22}. Philip et al.~\cite{Philip2010} also gave a quartic kernel for \textsc{POVD}, which Cygan 
  et al.~\cite{Cygan12} later improved to quadratic.  Our results are as follows.

  First, in Section~\ref{sec:pove}, we show that \pove admits a kernel of
  size~$O(k^2)$ and an algorithm with running time $O(4^km)$, thereby
  improving over the results of Ahmed et al.~\cite{Ahmed22} in a more general
  setting.

  Second, in Section~\ref{sec:povs}, we show that \povs has a kernel of size~$16k$ and it admits an algorithm with 
  running time $O((6k+12)^k \cdot m)$.  This answers the open question of Ahmed et al.~\cite{Ahmed22}.
  
  In Section~\ref{sec:tovs}, we consider analogous problem variants where the target is to obtain a graph of treewidth at most~1, rather than pathwidth at most~1. 
  Here we show that the deletion model and the explosion model are both equivalent to the problem \textsc{Feedback
  Vertex Set}, and that the split model is equivalent to \textsc{Feedback Edge Set} and can thus be solved in linear
  time. For the latter, Firbas~\cite{Firbas2023} independently obtained the same result.
  
  Finally, in Section~\ref{sec:minor-closed}, we consider the problem \SPLITPI (\splitPi), the generalized version of the splitting problem where the goal is to obtain a graph of a specific graph class $\Pi$ using at most $k$ split operations. 
  Eppstein et al.~\cite{Eppstein18} recently studied the similar problem of deciding whether a given graph $G$ is \emph{$k$-splittable}, i.e., whether it can be turned into a graph of $\Pi$ by splitting every vertex of $G$ at most $k$ times.
  For graph classes $\Pi$ that can be expressed in monadic second-order graph logic (MSO$_2$, see~\cite{Cygan15}), they gave an FPT algorithm parameterized by the solution size $k$ and the treewidth of the input graph.
  We adapt their algorithm for the problem \splitPi, resulting in an FPT algorithm parameterized by the solution size $k$ for MSO$_2$-definable graph classes $\Pi$ of bounded treewidth.
  Using a similar algorithm, we obtain the same result for the problem variant using vertex explosions.
   
  \section{Preliminaries}

  A parameterized problem $L$ with parameter $k$ is \emph{non-uniformly fixed-parameter tractable} if, for every value of $k$, there exists an algorithm that decides $L$ in time $f(k) \cdot n^{O(1)}$ for some computable function~$f$. 
  If there is a single algorithm that satisfies this property for all values of $k$, then $L$ is \emph{(uniformly) fixed-parameter tractable}.
  
  Given a graph $G$, we let $n$ and $m$ denote the number of vertices and edges of $G$, respectively. Since we can
  determine the subgraph of $G$ that contains no isolated vertices in $O(m)$ time, we assume, without loss of
  generality, that $n \in O(m)$. For a vertex $v \in V(G)$, we let $N(v) \coloneqq \{u \in V(G) \mid \adj(v, u)\}$
  and~$N[v] \coloneqq N(v) \cup \{v\}$ denote the open and closed neighborhood of $v$ in $G$, respectively.
  
  We refer to vertices of degree 1 as \emph{pendant} vertices.  For a vertex $v$ of $G$, we let ${\degs(v) \coloneqq 
  |\{u \in N(v) \mid \deg(u) > 1\}|}$ denote the degree of~$v$ ignoring its pendant neighbors. If $\degs(v) = d$, we 
  refer to $v$ as a vertex of \emph{degree*}~$d$. 
  A graph is a \emph{caterpillar} (respectively a 
  \emph{pseudo-caterpillar}), if it consists of a simple path (a simple cycle) with an 
  arbitrary number of adjacent pendant vertices. The path (the cycle) is called the \emph{spine} of the 
  \mbox{(pseudo-)}caterpillar.
  
  \begin{figure}[t]
    \centering
    \begin{subfigure}[b]{0.45\textwidth}
      \centering
      \includegraphics[page=1,scale=0.7]{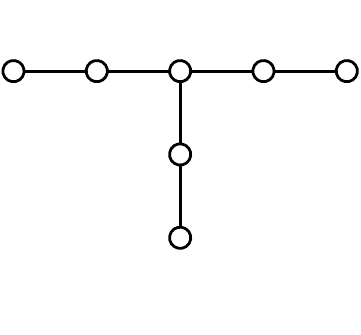}
      \caption{}
      \label{fig:subgraphsA}
    \end{subfigure}
    \hfill
    \begin{subfigure}[b]{0.45\textwidth}
      \centering
      \includegraphics[page=2,scale=0.7]{subgraphs}
      \caption{}
      \label{fig:subgraphsB}
    \end{subfigure}
    \caption{(a) The graph $T_2$. (b) Two graphs that do not contain $T_2$ as a subgraph, but both contain $N_2$ 
      (marked in orange) as a substructure.}
    \label{fig:subgraphs}
  \end{figure}
  
  Philip et al.~\cite{Philip2010} mainly characterized the graphs of pathwidth at most 1 as the graphs containing no 
  cycles and no $T_2$ (three simple paths of length~2 that all share an endpoint; see Figure~\ref{fig:subgraphsA}) as a subgraph. We additionally use slightly different sets of 
  forbidden substructures. An \emph{$N_2$ substructure} consists of a \emph{root} vertex $r$ adjacent to three distinct vertices of 
  degree at least~2. Note that every $T_2$ contains an $N_2$ substructure, 
  however, the existence of an $N_2$ substructure does not generally imply the existence of a $T_2$ subgraph; see 
  Figure~\ref{fig:subgraphsB}. In the following proposition, we state 
  the different characterizations for graphs of pathwidth at most 1 that we use in this work.
  \begin{myproposition}
    \label{prop:caterpillar}
    For a graph $G$, the following statements are equivalent.
    \begin{alphaenumerate}
      \item\mylabel{it:A} $G$ has pathwidth at most 1
      \item\mylabel{it:B} every connected component of $G$ is a caterpillar
      \item\mylabel{it:C} $G$ is acyclic and contains no $T_2$ subgraph
      \item\mylabel{it:D} $G$ is acyclic and contains no $N_2$ substructure
      \item\mylabel{it:E} $G$ contains no $N_2$ substructure and no connected component that is a pseudo-caterpillar.
    \end{alphaenumerate}
  \end{myproposition}
  \begin{proof}
    For the equivalences~\ref{it:A}$\iff$\ref{it:B}$\iff$\ref{it:C}, we refer to the paper by Philip et    
    al.~\cite{Philip2010}.
    
    We now show the equivalence~\ref{it:C}$\iff$\ref{it:D}. Since any $T_2$ subgraph also contains an $N_2$ 
    substructure, the implication~\ref{it:D} $\Rightarrow$~\ref{it:C} is clear. Consider a graph $G$ 
    that does not contain a cycle or a $T_2$ subgraph. Assume that $G$ contains an $N_2$ substructure, 
    i.e., a vertex $r$ with three neighbors $x$, $y$, and $z$ of degree at least~2. Let $a, b \in \{x,y,z\}$. Note that 
    $r \in N[a] \cap N[b]$. If $(N[a] \cap N[b]) \setminus \{r\} \neq \emptyset$, then $N[a] \cap N[b]$ contains a 
    cycle, a contradiction. Thus $N[a] \cap N[b] = \{r\}$, i.e., $x$, $y$, and $z$ are each adjacent to distinct 
    vertices of $V(G) \setminus \{r,x,y,z\}$. But then these vertices form a $T_2$ subgraph, a contradiction. Thus $G$ 
    contains no $N_2$ substructure.
    
    Finally, we show the equivalence~\ref{it:D}$\iff$\ref{it:E}. Since a pseudo-caterpillar contains a cycle as its 
    spine, the direction~\ref{it:D} $\Rightarrow$~\ref{it:E} is clear. Let $G$ be a graph containing no $N_2$ 
    substructures 
    or connected components that are a pseudo-caterpillar. Assume that $G$ contains a cycle $C$ and let $H$ denote the
    connected component containing $C$. Since $H$ contains no $N_2$ substructure and since every vertex of $C$ has two 
    other neighbors contained in $C$, all other vertices of $H$ must have degree 1 and are thus pendant vertices. 
    Therefore, $H$ is a pseudo-caterpillar, a contradiction. Thus $G$ contains no cycles and the implication~\ref{it:E} 
    $\Rightarrow$~\ref{it:D} follows.  
  \end{proof}

  We define the \emph{potential} of $v \in V(G)$ as $\mu(v) \coloneqq 
  \mathrm{max}(\degs(v) - 2, 0)$. The \emph{global potential} ${\mu(G) \coloneqq \sum_{v \in V(G)} \mu(v)}$ is defined 
  as the sum of the potentials of all vertices in $G$. Observe that $\mu(G) = 0$ if and only if $G$ contains no $N_2$ 
  substructure. The global potential thus indicates how far away we are from eliminating all $N_2$ substructures from 
  the instance.
  
  Recall that, for the problems \pove and \povs, the set $S \subseteq V(G)$ marks the vertices of $G$ that may be 
  chosen for the respective operations. We say that a set $W \subseteq S$ is a \emph{pathwidth-one explosion set} 
  (POES) of $G$, if the graph resulting from exploding all vertices in $W$ has pathwidth at most 1. Analogously, a 
  sequence of vertex splits on $S$ is a \emph{pathwidth-one split sequence} (POS-sequence), if the resulting graph has 
  pathwidth at most 1. We can alternatively describe a sequence of split operations as a \emph{split partition}, a function $\tau$ that maps 
  every vertex $v \in S$ to a partition of the edges incident to $v$. The number of splits corresponding to $\tau$ is 
  then defined by $|\tau| \coloneqq \sum_{v \in S} (|\tau(v)| - 1)$. We say that $|\tau|$ is the \emph{size} of $\tau$. 
  If $\tau$ corresponds to a POS-sequence, we refer to $\tau$ as a \emph{pathwidth-one split partition} (POS-partition).
  
  A graph class $\Pi$ is \emph{minor-closed} if, for every graph $G \in \Pi$ and for every minor $H$ of $G$, $H$ is also contained in $\Pi$.
  We say that a graph class $\Pi$ is \emph{MSO$_2$-definable}, if there exists an MSO$_2$ (monadic second-order graph logic, see~\cite{Cygan15}) formula $\varphi$ such that $G \models \varphi$ if and only if $G \in \Pi$.
  A graph class $\Pi$ has \emph{bounded treewidth} if there exists a constant $c \in \mathbb{N}$ such that every graph contained in $\Pi$ has treewidth at most~$c$. We let $\mathrm{tw}(\Pi)$ denote the minimum constant $c$ where this is the case.
  
  \section{FPT Algorithms for \POVE}
  \label{sec:pove}
  In this section, we first show that \pove can be solved in time $O(4^k \cdot m)$ using bounded search trees. 
  Subsequently, we develop a kernelization algorithm for \pove that yields a quadratic kernel in linear time.
  
   \subsection{Branching Algorithm}
  \label{sec:poveBranch}
  
  We start by giving a simple branching algorithm for \pove, similar to the algorithm by Philip et al.~\cite{Philip2010}
  for the deletion variant of the problem. For an $N_2$ substructure $X$, observe that exploding vertices not contained
  in $X$ cannot eliminate $X$, because the degrees of the vertices in $X$ remain the same due to the new degree-1
  vertices resulting from the explosion. To obtain a graph of pathwidth at most 1, it is therefore always necessary to
  explode one of the four vertices of every $N_2$ substructure by Proposition~\ref{prop:caterpillar}. Our branching rule
  thus first picks an arbitrary $N_2$ substructure from the instance and then branches on which of the four vertices of
  the $N_2$ substructure belongs to the POES. Recall that $S$ denotes the set of vertices of the input graph that can be
  exploded.
  \begin{branchrule}
    \label{br:pove}
    Let $r$ be the root of an $N_2$ substructure contained in $G$ and let $x$, $y$, and $z$ denote the three neighbors 
    of 
    $r$ in $N_2$. For every vertex $v \in \{r, x, y, z\} \cap S$, create a branch for the instance $(G', S \setminus 
    \{v\}, k - 1)$, where $G'$ is obtained from $G$ by exploding $v$.\\
    If $\{r, x, y, z\} \cap S = \emptyset$, reduce to a trivial no-instance instead.
  \end{branchrule}
  Note that an $N_2$ substructure can be found in $O(m)$ time by checking, for every vertex $v$ in~$G$, whether $v$ has 
  at least three neighbors of degree at least 2. Also note that vertex explosions do not increase the number of edges 
  of the graph. Since Branching Rule~\ref{br:pove} creates at most four new branches, each of which reduces 
  the parameter $k$ by 1, exhaustively applying the rule takes $O(4^k \cdot m)$ time. By 
  Proposition~\ref{prop:caterpillar}, it subsequently only remains to eliminate connected components that are a 
  pseudo-caterpillar. Since a pseudo-caterpillar can (only) be turned into a caterpillar by exploding a vertex of 
  its spine, the remaining instance can be solved in linear time.
  \begin{mytheorem}
    The problem \POVE can be solved in time $O(4^k \cdot m)$.
  \end{mytheorem}
  
  \subsection{Quadratic Kernel}
  \label{sec:poveKernel}
  We now turn to our kernelization algorithm for \pove. In this section, we develop a kernel of quadratic size, which 
  can be computed in linear time.
  
  We adopt our first two reduction rules from the kernelization of the deletion variant by Philip et 
  al.~\cite{Philip2010} and show that these rules are also safe for the explosion variant. The 
  first rule reduces the number of pendant neighbors of each vertex to at most one; see Figure~\ref{fig:RR1}.
    
  \begin{figure}[t]
    \centering
    \begin{subfigure}[b]{0.2\textwidth}
      \centering
      \includegraphics[page=1]{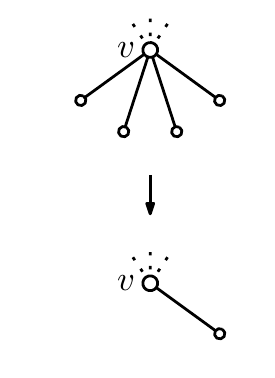}
      \caption{}
      \label{fig:RR1}
    \end{subfigure}
    \hfill
    \begin{subfigure}[b]{0.25\textwidth}
      \centering
      \includegraphics[page=2]{RR14}
      \caption{}
      \label{fig:RR2}
    \end{subfigure}
    \hfill
    \begin{subfigure}[b]{0.3\textwidth}
      \centering
      \includegraphics[page=3]{RR14}
      \caption{}
      \label{fig:RR3}
    \end{subfigure}
    \hfill
    \begin{subfigure}[b]{0.2\textwidth}
      \centering
      \includegraphics[page=4]{RR14}
      \caption{}
      \label{fig:RR4}
    \end{subfigure}
    \caption{ Examples for Reduction Rules~\ref{rr:pendant} (a), \ref{rr:caterpillar} (b), \ref{rr:pseudoCaterpillar} (c), and \ref{rr:nonAdj} (d).
      The vertices of $S$ are marked in green.}
  \end{figure}

  \begin{redrule}
    \label{rr:pendant}
    If $G$ contains a vertex $v$ with at least two pendant neighbors, remove all pendant neighbors of $v$ except one to
    obtain the graph $G'$ and reduce the instance to $(G',~{S \cap V(G')},~k)$.
  \end{redrule}
  \begin{proof}[Proof of Safeness]
    Observe that exploding a vertex of degree 1 has no effect, thus no minimum POES contains a vertex of degree 1. It 
    is therefore clear that any minimum POES of $G$ is also a POES of $G'$.
    
    Let $W$ denote a minimum POES of $G'$. It remains to show that $W$ is a POES of $G$. Let $l$ denote the 
    remaining pendant neighbor of $v$ in $G'$ and let $P \coloneqq V(G) \setminus V(G')$ denote the set of pendant 
    neighbors 
    of $v$ the reduction rule removed from $G$. Let 
    $\hat{G}$ and $\hat{G}'$ denote the graphs obtained by exploding the vertices of $W$ in $G$ and $G'$, respectively.
    If $v \in W$, then $\hat{G}$ only contains $|P|$ additional connected components compared to $\hat{G'}$, each of 
    which consists of two adjacent degree-1 vertices. Since $\hat{G'}$ has pathwidth at most 1 and a connected 
    component consisting of two adjacent degree-1 vertices also has pathwidth 1, $W$ is a POES of $G$.
    
    Now consider the case where $v \notin W$, i.e., $v \in V(\hat{G'})$. Recall that no minimum POES contains a vertex 
    of degree 1, thus $l \notin W$ and $v$ is still adjacent to $l$ in $\hat{G'}$. Since $\hat{G'}$ has pathwidth at 
    most 1, $\hat{G'}$ contains no cycles or $T_2$ subgraphs by Proposition~\ref{prop:caterpillar}. Note that the graph 
    $\hat{G}$ can be obtained from $\hat{G'}$ by adding the vertices of $P$ as pendant neighbors to~$v$, thus $\hat{G}$ 
    also contains no cycles. Since $v$ already 
    has a neighbor $l$ of degree 1 in $\hat{G'}$, adding additional pendant neighbors to $v$ does not introduce $T_2$ 
    subgraphs~\cite[Lemma 9]{Philip2010Full}. Hence $\hat{G}$ contains no $T_2$ subgraphs or cycles and thus $\hat{G}$ 
    has pathwidth at most one by Proposition~\ref{prop:caterpillar}. Therefore, $W$ is a POES of $G$.
  \end{proof}
  
  Since a caterpillar has pathwidth at most 1 by Proposition~\ref{prop:caterpillar}, we can safely remove any connected 
  component of $G$ that forms a caterpillar; see Figure~\ref{fig:RR2} for an example.
  
  \begin{redrule}
    \label{rr:caterpillar}
    If $G$ contains a connected component $X$ that is a caterpillar, remove $X$ from $G$ and reduce the instance
    to $(G-X,~{S \setminus V(X)},~k)$.
  \end{redrule}
  
  If $G$ contains a connected component that is a pseudo-caterpillar, then exploding an arbitrary vertex of its spine 
  yields a caterpillar. If the spine contains no vertex of $S$, the spine is a cycle that cannot be broken by a vertex 
  explosion. However, by Proposition~\ref{prop:caterpillar}, acyclicity is a necessary condition for a graph of 
  pathwidth at most 1. Hence we get the following reduction rule; see Figure~\ref{fig:RR3} for an illustration.

  \begin{redrule}
    \label{rr:pseudoCaterpillar}
    Let $X$ denote a connected component of $G$ that is a pseudo-caterpillar. If the spine of $X$ contains
    a vertex of $S$, remove $X$ from $G$ and reduce the instance to $(G-X,~{S \setminus V(X)},~k-1)$.
    Otherwise reduce to a trivial no-instance.
  \end{redrule}
  
  Recall that the degree* of a vertex is the number of its non-pendant neighbors. Our next goal is 
  to shorten paths of degree*-2 vertices to at most two vertices. If we have a path $x,y,z$ of degree*-2 vertices, we 
  refer to $y$ as a \emph{2-enclosed} vertex. Note that exploding a 2-enclosed vertex $y$ cannot eliminate any $N_2$ 
  substructures from the instance. By Proposition~\ref{prop:caterpillar}, vertex $y$ can thus only be part of an 
  optimal solution if exploding $y$ breaks cycles. 
  If we want to shorten the chain $x,y,z$ by contracting $y$ into one of its neighbors, we therefore need to ensure 
  that the shortened chain contains a vertex of $S$ if and only if the original chain contained a vertex of $S$. If $y 
  \in S$, we cannot simply add one of its neighbors, say $x$, to $S$ in the reduced instance, because exploding $x$ may 
  additionally remove an $N_2$ substructure; see Figure~\ref{fig:ShortenCounterexample} for an example. 
  While shortening paths of degree*-2 vertices to at most three vertices is simple, shortening them to length at most 2 (i.e., eliminating all 2-enclosed vertices) is therefore more involved.
  To solve this problem, we will show that we can greedily decide whether a 
  2-enclosed vertex $y$ is part of an optimal solution or not. This means that we can either immediately explode $y$, 
  or we can safely contract it into one of its degree*-2 neighbors. We start with the following auxiliary lemma.
  
  \begin{figure}[t]
    \centering
    \includegraphics[page=1,scale=1]{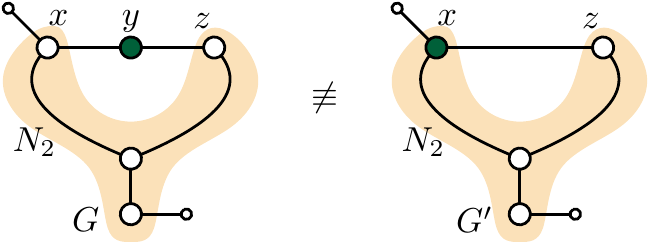}
    \caption{A graph $G$ that has no POES, because the highlighted $N_2$ substructure contains no vertex of $S$. For 
    the 
    graph $G'$ resulting from contracting $y$ into $x$, the set $\{x\}$ is a POES. The two instances are therefore not 
    equivalent.}
    \label{fig:ShortenCounterexample}
  \end{figure}
  
  \begin{mylemma}
    \label{lm:cycles}
    Let $y \in S$ be a 2-enclosed vertex of $G$ and let $\mathcal{C}_y$ denote the set of simple cycles of $G$ that 
    contain $y$. If $|C \cap S| \geq 2$ holds for every cycle $C \in \mathcal{C}_y$, then there exists a minimum POES 
    of $G$ that does not contain $y$.
  \end{mylemma}
  
  \begin{proof}
    Let $W$ be a minimum POES of size $k$ for $G$. As argued above, the 2-enclosed vertex $y$ can only be part of an 
    optimal solution if exploding $y$ breaks a cycle. Assume without loss of generality that $y \in W$.
    For two cycles $C_1, C_2 \in \mathcal{C}_y$, define $C_1 \oplus 
    C_2$ as the symmetric difference of the edges in $C_1$ and $C_2$, i.e., an edge is present in $C_1 \oplus C_2$ if 
    and only if it is present in exactly one of the cycles $C_1$ and $C_2$. Since $y$ is contained in a chain of 
    \mbox{degree*-2} vertices, $C_1 \oplus C_2$ is a collection of cycles that do not contain $y$. Now consider 
    the set $\hat{W} \coloneqq W \setminus \{y\}$. Since $W$ is a minimum POES for $G$ and since $C_1 \oplus 
    C_2$ does not contain $y$, exploding $\hat{W}$ still breaks all cycles of $C_1 \oplus C_2$. 
    
    Assume that there exist two distinct cycles $C_1$ and $C_2$ in $\mathcal{C}_y$ that both remain intact after 
    exploding $\hat{W}$. This means that all cycles of $C_1 \oplus C_2$ also remain intact, a contradiction to the 
    assumption that $W$ is a POES of $G$. We can therefore have at most one cycle $C \in \mathcal{C}_y$ with $C \cap 
    \hat{W} = \emptyset$. Since $|C \cap S| \geq 2$ holds by prerequisite of this lemma, we can pick an arbitrary 
    vertex $v \in C \cap S$ with $v \neq y$ and we find that $\hat{W} \cup \{v\}$ is a POES of size $k$ for $G$ that 
    does not contain vertex $y$.
  \end{proof}

  If a degree*-2 neighbor $y$ of a 2-enclosed vertex $v$ is contained in $S$, Lemma~\ref{lm:cycles} guarantees that 
  there exists a minimum POES that does not contain $v$, because every cycle that contains $v$ also contains $y$. We 
  can therefore define the following simple auxiliary reduction rule that ensures that no 2-enclosed vertex of $S$ is 
  adjacent to another degree*-2 vertex of~$S$; see Figure~\ref{fig:RR4}. This will be helpful for our next reduction 
  rule, because it reduces the number of cases we have to consider.
  \begin{redrule}
    \label{rr:nonAdj}
    Let $v$ be a 2-enclosed vertex of $G$ adjacent to a degree*-2 vertex $y \in S$. Reduce the instance to $(G, S 
    \setminus \{v\}, k)$.
  \end{redrule}
  Let $S_2 \subseteq S$ denote the set of 2-enclosed vertices contained in $S$. We now use Lemma~\ref{lm:cycles} to 
  greedily determine for each vertex in $S_2$ whether it should be contained in a minimum POES or not. Let $G' 
  \coloneqq G[V(G) \setminus (S \setminus S_2)]$ denote the graph obtained by removing all vertices of $S$ that are not 
  contained in $S_2$ from $G$. Let further $\hat{G}$ denote the graph obtained from $G'$ by first removing all pendant 
  vertices and subsequently contracting every connected component of $G'[V(G') \setminus S_2]$ into a single vertex. 
  Compute an arbitrary spanning forest $T$ of $\hat{G}$. 
  Let $\Sxpl \subseteq S_2$ denote the vertices of $S_2$ that are a leaf of $T$ and let $\Skeep \coloneqq 
  S_2 \setminus \Sxpl$ denote the remaining vertices of $S_2$ that are thus inner nodes of $T$; see 
  Figure~\ref{fig:spanningTree} for an illustration.

\begin{figure}
    \centering
    \begin{subfigure}[b]{0.3\textwidth}
      \centering
      \includegraphics[width=0.8\textwidth,page=1]{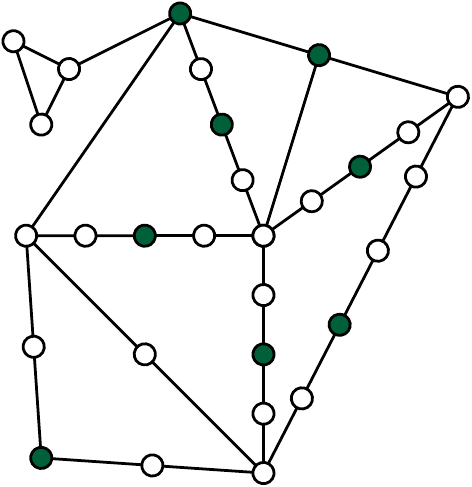}
      \caption{}
      \label{fig:spanningTreeA}
    \end{subfigure}
    \hfill
    \begin{subfigure}[b]{0.3\textwidth}
      \centering
      \includegraphics[width=0.8\textwidth,page=2]{spanningTree}
      \caption{}
      \label{fig:spanningTreeB}
    \end{subfigure}
    \hfill
    \begin{subfigure}[b]{0.3\textwidth}
      \centering
      \includegraphics[width=0.8\textwidth,page=3]{spanningTree}
      \caption{}
      \label{fig:spanningTreeC}
    \end{subfigure}
    \caption{(a) A graph $G$ with the vertices of $S$ marked in green. (b) The corresponding graph $G'$ obtained after 
    removing all vertices of $S$ that are not 2-enclosed, i.e., the remaining green vertices form the set $S_2$. (c) 
    The graph $\hat{G}$ obtained after contracting all connected components of $G'[V(G') \setminus S_2]$ into a single 
    vertex in $G'$, together with a spanning tree $T$ highlighted in orange. The vertices of $S_2$ that are leaves of 
    $T$ compose the set $\Sxpl$ (marked with red crosses), the remaining vertices of $S_2$ compose the set $\Skeep$.}
    \label{fig:spanningTree}
  \end{figure}

  \begin{mylemma}
    \label{lm:keepSplit}
    There exists a minimum POES $W$ of $G$ such that $\Skeep \cap W = \emptyset$ and $\Sxpl \subseteq W$.
  \end{mylemma}
  \begin{proof}
    Let $v \in \Skeep$ and let $C$ denote an arbitrary simple cycle of $G$ that contains $v$. We want to show 
    that $C$ always contains a vertex of $S \setminus \Skeep$, which will allow us to use Lemma~\ref{lm:cycles} to find 
    a minimum POES of $G$ that does not contain $v$. First consider 
    the case where $C$ is not completely contained in a connected component of $G'$. This means that removing the 
    vertices of $S \setminus S_2$ from $G$ splits cycle $C$, thus $C$ contains a vertex of $S \setminus S_2 \subseteq S 
    \setminus \Skeep$. Now consider the case where $C$ is completely contained in the graph $G'$. 
    If a vertex of $x \in C \cap S_2$ has degree 1 in $\hat{G}$, then $x$ is a leaf of $T$ and therefore 
    contained in $\Sxpl$, thus $C$ contains a vertex of $S \setminus \Skeep$. Otherwise every vertex of $C \cap S_2$ 
    has degree~2 in $\hat{G}$ and the construction of $\hat{G}$ ensures that cycle $C$ of $G'$ also induces a 
    cycle $\hat{C}$ in~$\hat{G}$. Because $T$ is a spanning forest of $\hat{G}$, cycle 
    $\hat{C}$ must contain an edge $e$ that is not contained in~$T$. Due to the construction of $\hat{G}$, one endpoint 
    $y$ of $e$ must be a vertex of $S_2$. But because $y$ has degree 2 in $\hat{G}$ and its incident edge $e$ is not 
    part of $T$, $y$ must be a leaf of $T$. Thus $y \in \Sxpl$, which again yields a vertex of $S \setminus \Skeep$ 
    contained in cycle $C$.
    
    We have shown that, for every $v \in \Skeep$ and for every simple cycle $C$ of $G$ containing $v$, $C$ also
    contains a vertex of $S \setminus \Skeep$. Consequently, by Lemma~\ref{lm:cycles}, there exists a minimum 
    POES of $G$ that does not contain $v$. Given the initial instance $\mathcal{I} = (G, S, k)$ of \pove, the instance 
    $\mathcal{I'} = (G, S \setminus \{v\}, k)$ is therefore equivalent. Because we have shown that every cycle of $G$ 
    that contains a vertex of $\Skeep$ also contains a vertex of $S \setminus \Skeep$, we can repeatedly apply this 
    step to obtain the equivalent instance $\mathcal{I^*} = (G, S \setminus \Skeep, k)$. Note that we do not actually 
    alter the initial instance $\mathcal{I}$, but the existence of the equivalent instance $\mathcal{I^*}$ shows that 
    there exists a minimum POES of $G$ that contains no vertices of $\Skeep$.
    
    Let $W$ denote a minimum POES of $G$ with $\Skeep \cap W = \emptyset$. We now want to show that 
    $\Sxpl \subseteq W$ always holds. First consider a vertex $v \in \Sxpl$ that has degree 1 
    in $\hat{G}$. Note that Reduction Rule~\ref{rr:nonAdj} ensures that $v$ is not adjacent to a degree*-2 vertex of 
    $S$ in $G$, thus 
    $v$ cannot have degree 1 in $G'$. Vertex $v$ can therefore only have degree 1 in $\hat{G}$ if both neighbors of $v$ 
    in $G'$ lie in the same connected component $H$ of $G'[V(G') \setminus S_2]$. Thus $v$ and vertices of $H$ form a 
    cycle $C$ with $C \cap S = \{v\}$. In order to break cycle $C$, it is therefore necessary to explode $v$ and thus 
    $v \in W$.
    
    Now consider a vertex $v \in \Sxpl$ that has degree 2 in $\hat{G}$. Let $x$ and $y$ denote the neighbors 
    of $v$ in $\hat{G}$. Let $\hat{C}$ denote the 
    cycle of $\hat{G}$ consisting of the path $x,v,y$ and the unique path of spanning forest $T$ connecting $x$ 
    and~$y$. Since $v$ is a leaf of $T$, one of the edges $vx$ or $vy$ is not contained in~$T$, thus $\hat{C}$ is 
    indeed a cycle. Reduction Rule~\ref{rr:nonAdj} ensures that no two vertices of $S_2$ are adjacent in~$G$, and 
    thus the construction of $\hat{G}$ guarantees that $x, y \notin S$. Because $x, y \notin S$, $v$ is the only vertex 
    of $S$ contained in $\hat{C}$ that is also a leaf 
    of $T$, thus $\hat{C} \cap \Sxpl = \{v\}$. We therefore only have a single vertex $v$ of $\Sxpl$ 
    contained in $\hat{C}$, all other vertices of $S$ contained in $\hat{C}$ must be a subset of $\Skeep$. Note that 
    the construction of $\hat{G}$ guarantees that we find a cycle $C$ in $G$ with the same properties. But 
    because the POES~$W$ does not contain any vertices of $\Skeep$, $W$ must contain $v$ in order to 
    break the cycle $\hat{C}$ in $\hat{G}$ (and consequently the cycle $C$ in $G$).
  \end{proof}

  Lemma~\ref{lm:keepSplit} now allows us to eliminate all 2-enclosed vertices and thus lets us shorten chains of 
  degree*-2 vertices to length at most 2. We state this in the following reduction rule; see Figure~\ref{fig:RR5} for 
  an illustration.
  
  \begin{figure}
    \centering
    \begin{subfigure}[b]{0.45\textwidth}
      \centering
      \includegraphics[width=\textwidth,page=2]{RRTwoEnclosed}
      \caption{}
      \label{fig:RR5}
    \end{subfigure}
    \hfill
    \begin{subfigure}[b]{0.45\textwidth}
      \centering
      \includegraphics[width=\textwidth,page=1]{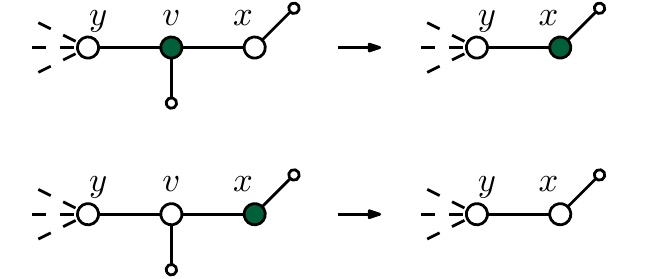}
      \caption{}
      \label{fig:RR6}
    \end{subfigure}
    \caption{Examples illustrating the two cases of Reduction Rule~\ref{rr:twoEnclosed} (a) and Reduction 
    Rule~\ref{rr:deg1}~(b).}
  \end{figure}
  
  \begin{redrule}
    \label{rr:twoEnclosed}
    Let $v$ denote a 2-enclosed vertex of $G$ with degree*-2 neighbors $x$ and~$y$.
    \begin{itemize}
      \item If $v \in \Sxpl$, let $G'$ denote the graph obtained from $G$ by exploding $v$. Reduce the instance to 
      $(G', S \setminus \{v\}, k - 1)$.
      \item Otherwise, remove $v$ from $G$, add the new edge $xy$, and reduce the instance to $(G - v + xy, S \setminus 
      \{v\}, k)$.
    \end{itemize}
  \end{redrule}
  \begin{proof}[Proof of Safeness]
    If $v \in \Sxpl$, then Lemma~\ref{lm:keepSplit} immediately tells us that it is safe to explode~$v$, thus the first 
    case is safe. 
    
    If $v \notin \Sxpl$, then Lemma~\ref{lm:keepSplit} lets us assume, without loss of generality, that $v \notin S$. 
    Note that $x$ and $y$ cannot be adjacent, because Reduction Rule~\ref{rr:pseudoCaterpillar} removes all connected 
    components that form a pseudo-caterpillar from $G$. The reduction therefore does not introduce multi-edges. Because $v$ is 
    2-enclosed, the reduction retains all $N_2$ substructures and all 
    cycles of~$G$. Because $v \notin S$, any solution for the original instance is also 
    a solution for the reduced instance and vice versa, thus the second case is also safe.
  \end{proof}
  
  To simplify the instance even further, the following reduction rule removes all degree*-2 vertices that are adjacent 
  to a vertex of degree* 1; see Figure~\ref{fig:RR6} for an illustration.
  
  \begin{redrule}
    \label{rr:deg1}
    Let $v$ be a degree*-2 vertex of $G$ with non-pendant neighbors $x$ and~$y$, such that $x$ has degree* 1.
    Remove $v$ from $G$ and add a new edge $xy$. If $v \in S$, reduce to $(G - v + xy, (S \setminus \{v\}) 
    \cup \{x\}, k)$. Otherwise reduce to $(G - v + xy, S \setminus \{x\}, k)$.
  \end{redrule}
  \begin{proof}[Proof of Safeness]
    Since $x$ has degree* 1, $x$ and $v$ cannot be contained in a cycle of $G$, and $x$ cannot be contained in a cycle 
    of $G - v + xy$. Hence we only have to consider $N_2$ substructures. Because $x$ itself has degree* 1 and is not 
    adjacent to a vertex of degree* at least~3, $x$ cannot be contained in an $N_2$ substructure of $G$. Since $x$ is 
    not contained in a cycle or $N_2$ substructure, $x$ is therefore also not contained in a minimum POES of $G$. Note 
    that $x$ must have a pendant neighbor, because otherwise, $x$ itself would be a pendant neighbor of~$v$. This means 
    that any $N_2$ substructure $H$ of $G$ containing $v$ is also present in $G - v + xy$, with vertex $x$ replacing 
    $v$ in $H$.
    Observe that the reduction modifies the set $S$ to ensure that $x$ can be exploded in the reduced instance 
    $\mathcal{I'}$ if and only if $v$ can be exploded in the original instance~$\mathcal{I}$. Therefore, a minimum POES 
    $W'$ of $\mathcal{I'}$ can be obtained from a minimum POES $W$ of $\mathcal{I}$ by replacing $v$ with 
    $x$ in $W$ and vice versa.
  \end{proof}
  
  Recall that the global potential $\mu(G)$ indicates how far away we are from our goal of eliminating all $N_2$ 
  substructures from $G$. With the following lemma, we show that our reduction rules ensure that the number of vertices 
  in the graph $G$ is bounded linearly in the global potential of $G$.
  
   \begin{mylemma}
    \label{lm:kernel}
    After exhaustively applying Reduction Rules~\ref{rr:pendant}--\ref{rr:deg1}, it holds that $|V(G)| \leq 8 \cdot 
    \mu(G)$.
  \end{mylemma}
  \begin{proof}
    Reduction Rule~\ref{rr:caterpillar} ensures that $G$ contains no vertices of degree* 0.
    For $i \in \{1,2\}$, let $V_i$ denote the set of non-pendant degree*-$i$ vertices of $G$ and let $V_3$
    denote the set of vertices with degree* at least 3.  Recall that we defined the global potential as
    \[\mu(G) = \sum_{v \in V(G)} \mu(v) = \sum_{v \in V(G)}\max(0, \degs(v) - 2).\] Since all vertices of $V_1$ and 
    $V_2$ have degree* at most 2, their potential is 0 and we get 
    \[ \mu(G) = \sum_{v \in V_3}(\degs(v) - 2) = \sum_{v \in V_3}\degs(v) - 2\cdot|V_3|. \]
    Note that $|V_3| \leq \mu(G)$, because each vertex of degree* at least 3 contributes at least 1 to the global 
    potential. We therefore get
    \begin{equation}
    \label{eq:3mu}
    \sum_{v \in V_3}\degs(v) \leq 3 \cdot \mu(G).
    \end{equation}
    By Reduction Rule~\ref{rr:twoEnclosed}, every vertex in $v \in V_2$ is adjacent to a vertex of~$V_1 \cup 
    V_3$, since otherwise, $v$ would be 2-enclosed. However, Reduction Rule~\ref{rr:deg1} additionally ensures that 
    vertices of $V_2$ cannot be adjacent to vertices of $V_1$, 
    thus every vertex of $V_2$ must be adjacent to a vertex of~$V_3$. Note that two adjacent vertices of $V_1$ would 
    form a caterpillar, which is prohibited by Reduction Rule~\ref{rr:caterpillar}. Therefore, every vertex of $V_1$ is 
    also adjacent to a vertex of $V_3$. 
    
    Overall, every vertex of $V_1$ and $V_2$ is thus adjacent to a vertex of $V_3$. Note that every vertex $v \in V_1$ 
    must additionally have a pendant neighbor, because otherwise, $v$ itself would be a pendant vertex. Hence every 
    vertex of $V_1$ and $V_2$ has degree at least 2 and thus contributes to the degree* of its neighbor in $V_3$. We 
    therefore have $|V_1| + |V_2| \leq  \sum_{v \in V_3}\degs(v)$, hence $|V_1| + |V_2| \leq 3\cdot \mu(G)$ by 
    Equation~\ref{eq:3mu}. Recall that $|V_3| \leq \mu(G)$, thus $|V_1| + |V_2| + |V_3| \leq 4\cdot \mu(G)$. By 
    Reduction Rule~\ref{rr:pendant}, each of these vertices can have at most one pendant neighbor and thus $|V(G)| \leq 
    8\cdot \mu(G)$.
  \end{proof}
  
  With Lemma~\ref{lm:kernel}, it now only remains to find an upper bound for the global potential~$\mu(G)$. We do this 
  using the following two reduction rules.
  \begin{redrule}
    \label{rr:explodeK}
    Let $v$ be a vertex of $G$ with potential $\mu(v) > k$. If $v \in S$, explode $v$ to obtain the graph $G'$ 
    and reduce the instance to $(G',~{S \setminus \{v\}},~k-1)$. Otherwise reduce to a trivial no-instance.
  \end{redrule}
  \begin{proof}[Proof of Safeness]
    Since exploding a vertex $u \in V(G)\setminus\{v\}$ decreases $\mu(v)$ by at most one, after exploding at most
    $k$ vertices in $V(G)\setminus\{v\}$ we still have $\mu(v) > 0$. Because $\mu(v) > 0$ implies that $G$ contains an 
    $N_2$ substructure, it is therefore always necessary to explode vertex~$v$ by Proposition~\ref{prop:caterpillar}.
  \end{proof}
  \begin{redrule}
    \label{rr:globalPOVE}
    If $\mu(G) > 2k^2 + 2k$, reduce to a trivial no-instance.
  \end{redrule}
  \begin{proof}[Proof of Safeness]
    By Reduction Rule~\ref{rr:explodeK} we have $\mu(v) \leq k$ and therefore $\degs(v) \leq k + 2$ for all $v \in 
    V(G)$. Hence
    exploding a vertex $v$ decreases the potential of $v$ by at most~$k$ and the potential of each of its non-pendant 
    neighbors by at most 1. Overall, $k$ vertex explosions can therefore only decrease the global potential $\mu(G)$ by 
    at most $k \cdot (2k + 2)$.
  \end{proof}

  Because Reduction Rule~\ref{rr:globalPOVE} gives us an upper bound for the global potential $\mu(G)$, we can use 
  Lemma~\ref{lm:kernel} to obtain the kernel.
  \begin{mytheorem}
    \label{thm:quadKernel}
    The problem \POVE admits a kernel of size $16k^2 + 16k$. It can be computed in time $O(m)$.
  \end{mytheorem}
  \begin{proof}
    By Reduction Rule~\ref{rr:globalPOVE}, using Lemma~\ref{lm:kernel} yields a kernel of size $16k^2 + 16k$ for 
    \pove. It remains to show that we can compute the kernel in linear time. 
    
    First observe that, while some reduction rules 
    may increase the number of vertices in the instance, the number of edges never increases.
    Also note that no reduction rule increases the global potential or the potential of a single vertex. We can 
    therefore apply Reduction Rules~\ref{rr:explodeK} and~\ref{rr:globalPOVE} exhaustively in the beginning in 
    $O(m)$ time. Subsequently, we use Reduction Rules~\ref{rr:caterpillar} 
    and~\ref{rr:pseudoCaterpillar} to eliminate all connected components that are caterpillars and pseudo-caterpillars, 
    respectively. To test for the latter, it suffices to check whether every vertex in the component has at most two
    neighbors of degree* 2 or higher, for the former it suffices to additionally test for acyclicity. Both reduction
    rules can thus be exhaustively applied in linear time. We then exhaustively apply Reduction Rules~\ref{rr:nonAdj}, 
    \ref{rr:twoEnclosed}, and \ref{rr:deg1} in linear time to eliminate most degree*-2 vertices. Since these three 
    rules only affect \mbox{degree*-2} vertices, each of them can be implemented using a single pass through the graph. 
    For Reduction Rule~\ref{rr:twoEnclosed}, it is 
    not hard to see that the auxiliary graphs $G'$ and $\hat{G}$, as well as the spanning tree $T$ used to determine 
    the sets $\Skeep$ and $\Sxpl$, can be computed in $O(m)$ time. Note that Reduction Rules~\ref{rr:caterpillar} 
    and~\ref{rr:pseudoCaterpillar} ensure that every connected component of $G$ contains an $N_2$ substructure. Since 
    Reduction Rules~\ref{rr:nonAdj}, \ref{rr:twoEnclosed}, and~\ref{rr:deg1} cannot eliminate $N_2$ substructures, no 
    connected component is a (pseudo-)caterpillar after applying Reduction Rules~\ref{rr:nonAdj}, 
    \ref{rr:twoEnclosed}, and~\ref{rr:deg1} and thus we do not 
    have to apply Reduction Rules~\ref{rr:caterpillar} and~\ref{rr:pseudoCaterpillar} again. 
    Finally, we use Reduction Rule~\ref{rr:pendant} to remove excess pendant neighbors at all vertices in linear time. 
    We therefore obtain the kernel in $O(m)$ time.
  \end{proof}
  
  \section{FPT Algorithms for \POVS}
  \label{sec:povs}  
  In this section, we first adapt the kernelization algorithm from Section~\ref{sec:poveKernel} to obtain a linear 
  kernel for \povs in linear time. Subsequently, we show that \povs can also be solved in time $O((6k + 12))^k \cdot 
  m)$ using bounded search trees.
  
  \subsection{Linear Kernel}
  \label{sec:povsKernel}
  
  In order to obtain a kernel for \povs, we reuse Reduction Rules~\ref{rr:pendant} -- \ref{rr:deg1}\footnote{In the 
  first case of Reduction Rule~\ref{rr:twoEnclosed}, instead of exploding the vertex $v$, split $v$ such that the two 
  non-pendant neighbors of $v$ are separated, thus breaking all cycles $v$ is contained in.} from
  Section~\ref{sec:poveKernel}. We first show that these reduction rules are also safe in the context of \povs.
  
  \begin{mylemma}
    Reduction Rules~\ref{rr:pendant} -- \ref{rr:deg1} are safe for the problem \POVS.
  \end{mylemma}
  \begin{proof}
    We first show that Reduction Rule~\ref{rr:pendant} is safe for \povs. Since $G'$ is an induced subgraph of $G$, it 
    is clear that any POS-partition of size $k$ for $G$ yields a POS-partition of size at most $k$ for $G'$.
    
    Conversely, let $\tau'$ be a POS-partition of size $k$ for $G'$. Let $l$ denote the 
    remaining pendant neighbor of $v$ in $G'$ and let $P \coloneqq V(G) \setminus V(G')$ denote the set of pendant 
    neighbors 
    of $v$ the reduction rule removed from $G$. Let $\tau$ denote the split partition of $G$ obtained from $\tau'$ by 
    adding all edges incident to the vertices in $P$ to the cell $c \in \tau'(v)$ containing the edge~$vl$. Note that 
    $\tau$ also has size $k$. Let $\hat{G}$ (respectively $\hat{G'}$) be the graph obtained from $G$ ($G'$) after 
    applying the splits defined by $\tau$ ($\tau'$). We want to show that $\hat{G}$ also has pathwidth at most 1. Let 
    $v_c'$ denote the vertex of $\hat{G'}$ corresponding to $c$ (i.e., $v_c'$ is adjacent to $l$) and let $v_c$ 
    denote the corresponding vertex of $\hat{G}$. Note that the only difference between $\hat{G'}$ and $\hat{G}$ is 
    that $v_c$ additionally has the vertices of $P$ as pendant neighbors. Since $\hat{G'}$ has pathwidth at most 1, 
    $\hat{G'}$ contains no $T_2$ subgraphs or cycles (Proposition~\ref{prop:caterpillar}). Because $\hat{G}$ only 
    contains additional degree-1 vertices, $\hat{G}$ also contains no cycles. Since $v_c'$ already has a pendant 
    neighbor $l$ in~$\hat{G'}$, adding additional pendant neighbors to $v_c'$ does not introduce any $T_2$ subgraphs 
    \cite[Lemma~9]{Philip2010Full}. Hence $\hat{G}$ contains no cycles and no $T_2$ subgraphs and thus has pathwidth at 
    most~1 by Proposition~\ref{prop:caterpillar}. Therefore, $\tau$ is a POS-partition of size $k$ for $G$ and we can 
    conclude that Reduction Rule~\ref{rr:pendant} is safe for \povs.
    
    Removing connected components that are caterpillars (Reduction Rule~\ref{rr:caterpillar}) is clearly also safe for 
    \povs. If a connected component $X$ of $G$ is a pseudo-caterpillar, any split that separates two edges belonging to 
    the spine of $X$ yields a caterpillar. Reduction Rule~\ref{rr:pseudoCaterpillar} is therefore also safe.
        
    As the next step, we show that Lemma~\ref{lm:cycles} also holds for minimum split sequences of \povs. Let $y \in S$ 
    be a 2-enclosed vertex of $G$ such that $|C \cap S| \geq 2$ holds for every simple cycle $C$ 
    containing $y$. Let $\phi$ denote a minimum POS-sequence for $G$ that splits $y$. Because $y$ is 2-enclosed, $y$ is 
    not contained in any $N_2$ substructures of $G$ and since a single split of $y$ can break all cycles containing 
    $y$, 
    $\phi$ splits $y$ exactly once. Let $\phi \setminus y$ denote the split sequence obtained from $\phi$ by 
    removing the split involving $y$. Using the same argument as the proof of Lemma~\ref{lm:cycles}, there is at most 
    one cycle $C$ in $G$ that is not broken by $\phi \setminus y$. Because $C$ contains a vertex $v \in S$ with $v \neq 
    y$, we can add an arbitrary split of $v$ that breaks $C$ to the sequence $\phi \setminus y$ and we obtain a 
    POS-sequence of size $k$ that does not split vertex $y$. The safeness of Reduction Rule~\ref{rr:nonAdj} again 
    immediately follows from Lemma~\ref{lm:cycles}.
    
    We now show that Lemma~\ref{lm:keepSplit} is also correct in the context of \povs. Specifically, we show that there 
    exists a minimum POS-sequence of $G$ that splits all vertices in $\Sxpl$ but no vertices of $\Skeep$. The first 
    part of the proof of Lemma~\ref{lm:keepSplit} uses Lemma~\ref{lm:cycles} to find a minimum POES that contains no 
    vertices of $\Skeep$. Since we have shown above that Lemma~\ref{lm:cycles} also holds for split operations, the 
    same strategy can be 
    used to find a minimum POS-sequence that splits no vertices of $\Skeep$. The second part of the proof shows that, 
    for every vertex $v \in \Sxpl$, there exists a cycle $C$ in $G$ that contains no other vertices of $S \setminus 
    \Skeep$. Given a POS-sequence $\phi$ that splits no vertices of $\Skeep$, it is therefore necessary that $\phi$ 
    splits $v$ in order to break cycle $C$, thus Lemma~\ref{lm:keepSplit} also holds for \povs. Since 
    Lemma~\ref{lm:keepSplit} is correct, the proof of safeness for Reduction Rule~\ref{rr:twoEnclosed} can also be 
    applied to \povs.
    
    Finally, consider Reduction Rule~\ref{rr:deg1}. Since $v$ is not contained in any cycles of $G$, $v$ is only 
    contained in a minimum POS-sequence $\phi$ of $G$ if $v$ is contained in an $N_2$ substructure of $G$. If $\phi$ 
    splits $v$, $\phi$ must therefore split off edge $yv$ alone at $v$, because otherwise, the resulting vertex still 
    has degree at least 2 and thus the $N_2$ substructure remains intact. Observe that, after splitting off $yv$ at 
    $v$, the other half of the split subsequently lies in a connected component that is a caterpillar. Since $\phi$ has 
    minimum size, $\phi$ therefore only splits $v$ once.
    Similarly, any minimum POS-sequence of $G'$ that splits $x$ must isolate the edge $yx$ and only splits $x$ once. 
    Analogously to the proof of Reduction Rule~\ref{rr:deg1} in Section~\ref{sec:poveKernel}, we can therefore obtain a 
    minimum POS-sequence $\phi'$ of $G'$ from a minimum POS-sequence $\phi$ of $G$ by replacing the uniquely defined 
    split of $v$ in $\phi$ with the uniquely defined split of $x$ and vice versa.
  \end{proof}
  
  As in Section~\ref{sec:poveKernel}, we now define a reduction rule that gives an upper bound for the global potential 
  $\mu(G)$. To obtain this upper bound, it suffices to show that a single split operation can decrease the global potential by at most 2.
  \begin{redrule}
    \label{rr:globalPOVS}
    If $\mu(G) > 2k$, reduce to a trivial no-instance.
  \end{redrule}
  \begin{proof}[Proof of Safeness]
    Consider a vertex $v$ of $G$ being split into two new vertices $v_1$ and $v_2$. We show that the global potential 
    decreases by at most 2.
    
    If $\degs(v_1) = 0$, then $\mu(v_1) = 0$ and thus $\mu(v_1) + \mu(v_2) = \mu(v_2) = \mu(v)$. Additionally, all 
    neighbors of $v_1$ are pendant vertices whose potential remains unchanged. Note that the potential of neighbors
    of $v_2$ can only have decreased (by at most 1) if $\deg(v_2) = 1$. Overall, the global potential thus decreases by 
    at most 1 if $\degs(v_1) = 0$.
    
    If $\degs(v_1) = 1$ and $\degs(v_2) = 1$ then $\mu(v) = 0$ and the potential of the two non-pendant neighbors 
    decreases by at most 1 each.
    
    If $\degs(v_1) = 1$ and $\degs(v_2) \geq 2$, then the potential of the 
    non-pendant neighbor of $v_1$ decreases 
    by at most 1 and $\mu(v_1) + \mu(v_2) = \mu(v_2) = \mu(v) - 1$, hence the global potential decreases by at most 2.
    
    Finally, if $\degs(v_1) \geq 2$ and $\degs(v_2) \geq 2$, then the potential of the neighbors of $v$ does not change 
    and $\mu(v_1) + \mu(v_2) = \degs(v_1) - 2 + \degs(v_2) - 2 = \degs(v) - 4 = \mu(v) - 2$, hence the global potential
    decreases by 2. 
    
    Note that all remaining cases are symmetric. Therefore, a single split can decrease the global potential by at most 
    2. Since $\mu(G) > 0$ implies that $G$ contains an $N_2$ substructure, any instance with $\mu(G) > 2k$ is therefore 
    a no-instance by Proposition~\ref{prop:caterpillar}.
  \end{proof}
  
  Because Reduction Rule~\ref{rr:globalPOVS} gives a linear upper bound for the global potential of $G$, we can use 
  Lemma~\ref{lm:kernel} from Section~\ref{sec:poveKernel} to obtain a linear kernel for \povs.
  
  \begin{mytheorem}
    \label{thm:linKernel}
    The problem \POVS admits a kernel of size~$16k$. It can be computed in time $O(m)$.
  \end{mytheorem}
  \begin{proof}
    After exhaustively applying Reduction Rules~\ref{rr:pendant} -- \ref{rr:deg1}, using Lemma~\ref{lm:kernel} with the 
    upper bound $\mu(G) \leq 2k$ provided by Reduction Rule~\ref{rr:globalPOVS} yields a kernel of size $16k$ for \povs.
    
    To obtain this kernel in linear time, we first apply Reduction Rule~\ref{rr:globalPOVS} once in the beginning. The 
    proof of Theorem~\ref{thm:quadKernel} shows that the remaining Reduction Rules~\ref{rr:pendant} -- \ref{rr:deg1} 
    can be applied exhaustively in time $O(m)$. 
  \end{proof}
  
  \subsection{Branching Algorithm}
  
  We now propose an alternative FPT algorithm for \povs using bounded search trees. We reuse Reduction 
  Rule~\ref{rr:pseudoCaterpillar} to eliminate connected components that are pseudo-caterpillars. Similar to 
  Section~\ref{sec:poveBranch}, our branching rule will remove all $N_2$ substructures from the instance. For the 
  vertex split operation, however, we need to additionally consider the possible ways to split a single vertex.  
  The following lemma helps us limit the number of suitable splits and thus decreases the size of our branching vector.
  
  \begin{mylemma}
    \label{lm:twoSplit}
    For every instance of \povs, there exists a minimum POS-sequence $\phi$ such that every split operation in $\phi$ 
    splits off at most two edges.
  \end{mylemma}
  \begin{proof}
    Consider a minimum POS-partition $\tau$ of $G$. In order to prove the statement of 
    the lemma, we want to show that we can alter $\tau$ such that, for every $v \in S$, $\tau(v)$ contains at most one 
    cell with more than two elements. Let $G'$ denote the graph obtained from $G$ by applying the splits corresponding 
    to $\tau$ to $G$, thus $G'$ has pathwidth at most 1 and contains no cycles or $N_2$ substructures by 
    Proposition~\ref{prop:caterpillar}. For a cell $c \in \tau(v)$, let $v_c$ denote the corresponding vertex of $G'$. 
    Note that $v_c$ can have at most two neighbors of degree 2 or higher, because otherwise, we find an $N_2$ 
    substructure in $G'$, a contradiction. 
    All other neighbors of $v_c$ must therefore be degree-1 vertices. Now fix an arbitrary cell $c \in 
    \tau(v)$ with $|c| \geq 2$ (if no such cell exists, we are already done). For every cell $c' \in \tau(v)$ with $c' 
    \neq c$, we move $\mathrm{max}(0, |c'| - 2)$ edges corresponding to degree-1 vertices in $G'$ from cell $c'$ to 
    cell $c$ in $\tau(v)$ (and therefore from vertex $v_{c'}$ to vertex $v_c$ in $G'$). Since the vertex $v_c$ has 
    degree 2 or higher, adding additional degree-1 neighbors to it does not increase the pathwidth of $G'$ 
    (\cite[Lemma~10]{Philip2010Full}). Similarly, removing the degree-1 vertices from the other vertices also does not 
    increase the pathwidth. 
    This leads to a partition of size $|\tau(v)|$ for the edges 
    incident to $v$ such that at most one cell contains more than two edges. Since this partition can be realized by a 
    sequence splitting off at most two edges per operation, this concludes the proof.
  \end{proof}
  
  In addition to Reduction Rule~\ref{rr:pseudoCaterpillar}, we also reuse Reduction Rule~\ref{rr:pendant} to limit the 
  number of pendant neighbors for each vertex, and Reduction Rule~\ref{rr:globalPOVS} to bound the global 
  potential~$\mu(G)$. These two rules together bound the degree of all vertices in $G$, which lets us state the 
  following branching rule; see Figure~\ref{fig:branching} for an illustration.
  
    \begin{figure}
    \centering
    \begin{subfigure}[b]{0.3\textwidth}
      \centering
      \includegraphics[page=1]{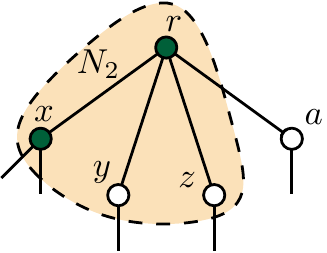}
      \caption{}
    \end{subfigure}
    \hfill
    \begin{subfigure}[b]{0.3\textwidth}
      \centering
      \includegraphics[page=2]{BranchingRule}
      \caption{}
    \end{subfigure}
    \hfill
    \begin{subfigure}[b]{0.3\textwidth}
      \centering
      \includegraphics[page=3]{BranchingRule}
      \caption{}
    \end{subfigure}
    \caption{(a) An $N_2$ substructure consisting of the vertices $\{r, x, y, z\}$. (b)-(c) Two of the branches of 
    Branching Rule~\ref{br:povs} eliminating the $N_2$ substructure. The former splits off edge $rx$ at $x$, the latter 
    splits off the edges $rz$ and $ra$ at $r$.}
    \label{fig:branching}
  \end{figure}

  \begin{branchrule}
    \label{br:povs}
    Let $r$ be the root of an $N_2$ substructure contained in $G$ and let $x$, $y$, and $z$ denote the three neighbors 
    of $r$ in $N_2$. If $\{r,x,y,z\} \cap S = \emptyset$, reduce to a trivial no-instance. Otherwise branch on the 
    following instances:
    \begin{itemize}
      \item For every $v \in \{x, y, z\} \cap S$, create a separate branch for the instance $(G', S', k - 
      1)$, where $G'$ is the graph obtained from $G$ by splitting off the edge $rv$ at $v$.
      \item If $r \in S$: For every subset $X \subseteq N(v)$ with $|X| \leq 2$ and $|X \cap \{x, y, z\}| \geq 
      1$, create a separate branch for the instance $(G', S', k - 1)$, where $G'$ is obtained from $G$ by splitting off 
      the edges corresponding to $X$ at $r$.
    \end{itemize}
    We set $S' \coloneqq S \cup K$ in all branches, where $K$ is the set of vertices created by the split operation.
  \end{branchrule}
  
  \begin{mylemma}
    \label{lm:branchVector}
    Exhaustively applying Reduction Rules~\ref{rr:pendant} and \ref{rr:globalPOVS} and Branching 
    Rule~\ref{br:povs} yields an equivalent instance without $N_2$ substructures in time ${O((6k+12)^k \cdot m)}$.
  \end{mylemma}
  \begin{proof}
    Let $T$ denote the $N_2$ substructure induced by $r$ and its neighbors $\{x,y,z\}$. Observe that $T$ can only be 
    removed from the graph by splitting one of the vertices in $\{r,x,y,z\}$. If $T$ contains no vertex of $S$, we 
    consequently have a no-instance.
    
    If we split a vertex $v \in \{x,y,z\}$ in 
    order to eliminate $T$, then we can only split off the edge $rv$ alone, because otherwise, the resulting vertex 
    adjacent to $r$ still has degree 2. We thus only need three branches to enumerate all suitable splits of the 
    vertices 
    in $\{x,y,z\}$.
    
    If we split vertex $r$, then splitting off any subset of $N(r)$ that contains one or two vertices of $\{x,y,z\}$ 
    is suitable to eliminate $T$. However, by Lemma~\ref{lm:twoSplit}, it suffices to consider subsets of $N(r)$ 
    with size at most~2. 
    By Reduction Rule~\ref{rr:globalPOVS}, the global potential $\mu(G)$ is at most $2k$, thus $r$ can have at most $2k 
    + 2$ neighbors of degree 2 or higher. Since Reduction Rule~\ref{rr:pendant} limits the number of degree-1 neighbors 
    of $r$ to at most one, $r$ has degree at most $2k + 3$. We therefore need at most $3 \cdot (2k + 3)$ branches to 
    enumerate all subsets of $N(r)$ of size at most 2 containing at least one vertex of $\{x,y,z\}$. Together with the 
    three branches from earlier, this yields $6k+12$ branches, each of which reduces the parameter by 1.
    
    We have shown earlier that all of our reduction rules can be applied exhaustively in linear time. Finding an $N_2$ 
    substructure in Branching Rule~\ref{br:povs} can also be achieved in $O(m)$ time by checking, for each vertex $v$ 
    in $G$, whether $v$ has at least three neighbors of degree 2 or higher. We thus find an equivalent instance without 
    $N_2$ substructures in time $O((6k+12)^k \cdot m)$.
  \end{proof}
  
  By Lemma~\ref{lm:branchVector}, Branching Rule~\ref{br:povs} eliminates all $N_2$ substructures from the graph. 
  Reduction Rule~\ref{rr:pseudoCaterpillar} additionally removes all pseudo-caterpillars from the graph, we therefore 
  obtain a graph of pathwidth at most 1 by Proposition~\ref{prop:caterpillar}.
  \begin{mytheorem}
    The problem \POVS can be solved in time $O((6k + 12)^k \cdot m)$.
  \end{mytheorem}
  
  \section{Treewidth-One Vertex Splitting}
  \label{sec:tovs}
  In this section, we consider the variant of \povs where the goal is to obtain a graph of treewidth at most 1, rather 
  than pathwidth at most 1. We remark that a graph $G$ has treewidth at most 1 if and only if $G$ is a forest.
  \begin{problem}
    \problemtitle{\TOVS (\tovs)}
    \probleminput{An undirected graph $G = (V, E)$, a set $S \subseteq V$, and a positive integer~$k$.}
    \problemquestion{Is there a sequence of at most $k$ splits on vertices in $S$ such that the resulting
      graph has treewidth at most 1?}
  \end{problem}

  Note that the variant of \tovs with the deletion operation is exactly the problem \textsc{Feedback Vertex Set}, which 
  is a well-studied NP-complete \cite{Karp72} problem that admits a quadratic kernel~\cite{Thomasse10}. Also note that, in this setting, removing 
  degree-1 vertices from the graph yields an equivalent instance. For this reason, the variant with the explosion 
  operation is also equivalent to \textsc{Feedback Vertex Set}.
  We thus only focus on the problem \tovs, for which we give a simple linear-time algorithm.
  Analogously to POS-sequences and POS-partitions, we define TOS-sequences and TOS-partitions as split sequences and split partitions, respectively, that result in a graph of treewidth at most 1.
  
  \begin{mylemma}
    \label{lm:minTOSSeq}
    Every minimum TOS-sequence of a graph $G$ has size $|E(G)| - |V(G)| + 1$.
  \end{mylemma}
  \begin{proof}
    We assume without loss of generality that $G$ is connected.
    Consider a minimum TOS-partition $\tau$ of size $k$ for $G$ and let $G'$ be the graph resulting from $\tau$.
    Assume that $G'$ is disconnected.
    Then there exists a vertex $v$ and two distinct cells $c_1, c_2 \in \tau(v)$, such that the vertices $v_{c_1}$ and $v_{c_2}$ are not connected in $G'$.
    Since $v_{c_1}$ and $v_{c_2}$  are not connected, merging them into a single vertex does not introduce any cycles in $G'$.
    We can thus merge $c_1$ and $c_2$ into a single cell in $\tau(v)$ and we obtain a TOS-partition of size $k - 1$ for $G$, a contradiction to the minimality of $\tau$.
    Therefore, for any minimum TOS-sequence of $G$, the resulting graph~$G'$ must be connected and is thus a tree with $|E(G')| = |V(G')| - 1$.
    Since a single split operation increases the number of vertices by exactly 1 and does not alter the number of edges, it is $|E(G')| =|E(G)|$ and $|V(G')| = |V(G)| +k$ and thus $k = |E(G)| - |V(G)| + 1$.
  \end{proof}

  Note that a graph $G$ with a set $S$ defining its splittable vertices has a TOS-sequence if and only if $G[V(G) \setminus S]$ is acyclic.
  Together with Lemma~\ref{lm:minTOSSeq}, it thus follows that an instance $(G, S, k)$ of \tovs is a yes-instance if and only if $G[V(G) \setminus S]$ is acyclic and $k \geq |E(G)| - |V(G)| + 1$. 
  Since the acyclicity of a graph can be tested in linear time using a simple depth-first search, we obtain the
  following result, which was also independently shown by Firbas~\cite{Firbas2023}.

  \begin{mytheorem}
    The problem \tovs can be solved in time $O(n + m)$.
  \end{mytheorem}

  In fact, Lemma~\ref{lm:minTOSSeq} implies that the problem of determining whether a graph can be turned into a forest
  using at most $k$ splits is equivalent to the problem \textsc{Feedback Edge Set}, which asks whether a given graph can be turned into a forest using at most $k$ edge deletions.
  
  \section{FPT Algorithms for Splitting and Exploding to MSO$_2$-Definable Graph Classes of Bounded Treewidth}
  \label{sec:minor-closed}
  
  While the previous sections focused on the problems of obtaining graphs of pathwidth and treewidth at most 1, respectively, using at most $k$ vertex splits or explosions on the input graph, we now consider the problem of obtaining other graph classes using these operations.
  With the following problems, we generalize the problems from the previous sections.
  
  \begin{problem}
    \problemtitle{\SPLITPI (\splitPi)}
    \probleminput{An undirected graph $G = (V, E)$, a set $S \subseteq V$, and a positive integer~$k$.}
    \problemquestion{Is there a sequence of at most $k$ splits on vertices in $S$ such that the resulting
      graph is contained in $\Pi$?}
  \end{problem}

  \begin{problem}
    \problemtitle{\EXPLODEPI(\explodePi)}
    \probleminput{An undirected graph $G = (V, E)$, a set $S \subseteq V$, and a positive integer~$k$.}
    \problemquestion{Is there a set~$W \subseteq S$ with~$|W| \le k$ such that the graph resulting from exploding all vertices in~$W$ is contained in $\Pi$?}
  \end{problem}
  
  In the following, we show that \splitPi and \explodePi are both FPT parameterized by the solution size $k$, if the graph class $\Pi$ is MSO$_2$-definable and has bounded treewidth. We first consider the split operation because here we can use results from related problems.
  
  \subsection{Vertex Splitting}
  \label{sec:splitMSO}
  Nöllenburg et al. \cite{Nollenburg22} showed that, for any minor-closed graph class $\Pi$, the graph class $\Pi_k$ containing all graphs that can be modified to a graph in $\Pi$ using at most $k$ vertex splits is also minor-closed.
  Robertson and Seymour \cite{Robertson95} showed that every minor-closed graph class has a constant-size set of forbidden minors and that it can be tested in cubic time whether a graph contains a given fixed graph as a minor.
  Since $\Pi_k$ is minor-closed, this implies the existence of a non-uniform FPT-algorithm for the problem \splitPi.
  
  \begin{myproposition}[\cite{Nollenburg22}]
    \label{prop:nonUniform}
    For every minor-closed graph class $\Pi$, the problem \splitPi is non-uniformly FPT parameterized by the solution size $k$.
  \end{myproposition}
  
  In the following, we show that the problem \splitPi is uniformly FPT parameterized by $k$ if $\Pi$ is MSO$_2$-definable and has bounded treewidth. Since every minor-closed graph class is MSO$_2$-definable~\cite{Robertson95}, this improves the result from Proposition~\ref{prop:nonUniform} for graph classes of bounded treewidth.

  Eppstein et al.~\cite{Eppstein18} showed that the problem of deciding whether a given graph $G$ can be turned into a graph of class $\Pi$ by splitting each vertex of $G$ at most $k$ times can be expressed as an MSO$_2$ formula on $G$, if $\Pi$ itself is MSO$_2$-definable. Using Courcelle's Theorem~\cite{Courcelle90}, this yields an FPT-algorithm parameterized by $k$ and the treewidth of the input graph. 
  Their algorithm exploits the fact that the split operations create at most $k$ copies of each vertex in the graph.
  Since the same also applies for the problem \splitPi, where we may apply at most $k$ splits overall, their algorithm can be straightforwardly adapted for \splitPi, thereby implying the following result.
  \begin{mycorollary}
    \label{cor:fptKTreewidth}
    For every MSO$_2$-definable graph class $\Pi$, the problem $\splitPi$ is FPT parameterized by the solution size $k$ and the treewidth of the input graph.
  \end{mycorollary}

  For a graph class $\Pi$ of bounded treewidth, recall that $\mathrm{tw}(\Pi)$ denotes the maximum treewidth among all graphs in $\Pi$. 
  With the following lemma, we show that, if the target graph class $\Pi$ has bounded treewidth, then every yes-instance of \splitPi must also have bounded treewidth.
  \begin{myproposition}
    \label{prop:boundedTreewidth}
    For a graph class $\Pi$ of bounded treewidth, let $\mathcal{I} = (G, S, k)$ be an instance of \splitPi . If $\mathrm{tw}(G) > k +  \mathrm{tw}(\Pi)$, then $\mathcal{I}$ is a no-instance.
  \end{myproposition}
  \begin{proof}
    We first show that a single split operation can reduce the treewidth of $G$ by at most~1.
    Assume, for the sake of contradiction, that we can obtain a graph $G'$ of treewidth less than $\mathrm{tw}(G) - 1$ 
    by splitting a single vertex $v$ of $G$ into vertices $v_1$ and $v_2$ of $G'$. Let~$\mathcal{T}$ denote a minimum tree 
    decomposition of $G'$. Remove all occurences of $v_1$ and $v_2$ in $\mathcal{T}$ and add $v$ to every bag of 
    $\mathcal{T}$. Observe that the result is a tree decomposition of size less than $\mathrm{tw}(G)$ for~$G$, a 
    contradiction.  A single split operation thus decreases the treewidth of the graph by at most 1.
    Since every graph $G' \in \Pi$ has $\mathrm{tw}(G') \leq \mathrm{tw}(\Pi)$, it is thus impossible to obtain a graph of $\Pi$ with at most $k$ vertex splits if $\mathrm{tw}(G) > k +  \mathrm{tw}(\Pi)$.
  \end{proof}

  Given a graph class $\Pi$ of bounded treewidth, we first determine in time $f(k + \mathrm{tw}(\Pi)) \cdot n$ 
  whether the treewidth of $G$ is greater than $k + \mathrm{tw}(\Pi)$~\cite{Bodlaender93}. If 
  this is the case, then we can immediately report a no-instance by Proposition~\ref{prop:boundedTreewidth}.
  Otherwise, we know that $\mathrm{tw}(G) \leq k + \mathrm{tw}(\Pi)$. Since $\mathrm{tw}(\Pi)$ is a constant, we have $\mathrm{tw}(G) \in O(k)$, and thus Corollary~\ref{cor:fptKTreewidth}  yields the following result.

  \begin{mytheorem}
    \label{thm:splitPiFPT}
    For every MSO$_2$-definable graph class $\Pi$ of bounded treewidth, the problem $\splitPi$ is FPT parameterized by the solution size $k$.
  \end{mytheorem}
  
  \subsection{Vertex Explosion}
  We now turn to the problem variant \explodePi that uses vertex explosions instead of vertex splits.
  Analogously to Section~\ref{sec:splitMSO}, we let $\PiExpl$ denote the graph class containing all graphs that can be modified to a graph in $\Pi$ using at most $k$ vertex explosions.
  For arbitrary minor-closed graph classes $\Pi$, the class $\PiExpl$ is not necessarily minor-closed, as the counterexample in Figure~\ref{fig:ExplosionCounterexample} shows.
  It is therefore not clear whether Proposition~\ref{prop:nonUniform} also holds for \explodePi.
  Note that, in Figure~\ref{fig:ExplosionCounterexample}, splitting off a single edge in $H_1$ yields a graph of $\Pi$.
  The question whether a graph of $\Pi$ can be obtained by applying arbitrarily many vertex splits to at most $k$ vertices in the input graph is therefore not equivalent to \explodePi for arbitrary graph classes $\Pi$.
  
  \begin{figure}[t]
    \centering
    \begin{subfigure}[b]{0.45\textwidth}
      \centering
      \includegraphics[page=1,scale=1]{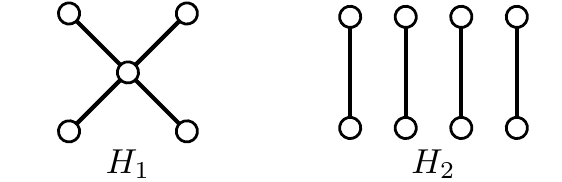}
      \caption{}
    \end{subfigure}
    \hfill
    \begin{subfigure}[b]{0.45\textwidth}
      \centering
      \includegraphics[page=2,scale=1]{ExplosionCounterexample}
      \caption{}
    \end{subfigure}
    \caption{(a) Two forbidden minors $H_1$ and $H_2$ characterizing a minor-closed graph class $\Pi$. (b) A graph $G \notin \Pi$ that can be modified into the graph $G' \in \Pi$ by exploding vertex $x$, thus $G \in \Pi_1^\times$. However, exploding at most one vertex in the graph $H_1$ yields either $H_1$ or $H_2$, thus $H_1 \notin \Pi_1^\times$. Since $H_1$ is a minor of $G$, $\Pi_1^\times$ is therefore not minor-closed.}
    \label{fig:ExplosionCounterexample}
  \end{figure}

  Additionally, the FPT-algorithm for \splitPi derived from Eppstein et al.~\cite{Eppstein18} cannot be straightforwardly adapted for \explodePi, since the number of new vertices resulting from explosions is not bounded by a function in $k$.
  However, we use a similar approach for \explodePi by defining an MSO$_2$ formula on an auxiliary graph, again yielding an FPT algorithm parameterized by the solution size $k$ for MSO$_2$-definable graph classes $\Pi$ of bounded treewidth.
  
  Given an instance $(G, S, k)$ of \explodePi, we first construct the auxiliary graph $G^\times = (V(G) \cup D, E')$ by subdividing each edge of $G$ twice; see Figure~\ref{fig:MSO2Exp}. The vertices of $D$ denote the new subdivision vertices.
  The subdivision vertices adjacent to a vertex $v \in V(G)$ in $G^\times$ represent the vertices that result from exploding $v$; see Figure~\ref{fig:MSO2ExpC}.
  
  \begin{figure}[t]
    \centering
    \begin{subfigure}[b]{0.32\textwidth}
      \centering
      \includegraphics[page=1,scale=1]{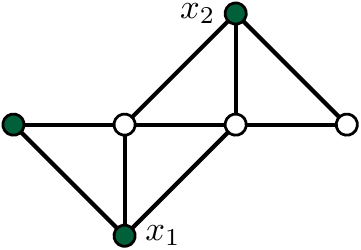}
      \caption{}
      \label{fig:MSO2ExpA}
    \end{subfigure}
    \hfill
    \begin{subfigure}[b]{0.32\textwidth}
      \centering
      \includegraphics[page=2,scale=1]{MSO2Explosion}
      \caption{}
      \label{fig:MSO2ExpB}
    \end{subfigure}
    \hfill
    \begin{subfigure}[b]{0.32\textwidth}
      \centering
      \includegraphics[page=3,scale=1]{MSO2Explosion}
      \caption{}
      \label{fig:MSO2ExpC}
    \end{subfigure}
    \caption{(a) An instance $(G, S, 2)$ of \explodePi. (b) The 
    corresponding auxiliary graph $G^\times$ obtained by subdividing each edge in $G$ twice. (c) The graph obtained by exploding $\{x_1, x_2\}$ in $G$ is the highlighted minor of $G^\times$. Since $\Pi$ is MSO$_2$-definable, one can express \explodePi using an MSO$_2$ formula on $G^\times$.}
    \label{fig:MSO2Exp}
  \end{figure}
  
  Given a set  $W \subseteq S \subseteq V(G)$ representing the vertices of $G$ that are chosen to be exploded, our MSO$_2$ formula on $G^\times$  works as follows.
  The graph that is obtained from $G$ by exploding the vertices of $W$ is exactly the graph $G^\times_W$ obtained from $G^\times$ by removing all vertices of $W$ and by contracting all subdivision vertices of $D$ adjacent to a vertex $v \in V(G) \setminus W$ into $v$; see Figure~\ref{fig:MSO2ExpC} for an example.
  We thus simply need to test whether the minor $G^\times_W$ of $G^\times$ is contained in $\Pi$.
  
  Let $\Pi$ be an MSO$_2$-definable graph class and let $\varphi$ denote the corresponding MSO$_2$-formula such that $G^\times \models \varphi$ if and only if $G^\times$ is contained in $\Pi$. We let $V^\times \coloneqq V(G^\times)$ and $E^\times \coloneqq E(G^\times)$ denote the free variables of $\varphi$ that correspond to the vertices and edges of~$G^\times$, respectively.
  In order to test whether the minor $G^\times_W$ of $G^\times$  is contained in $\Pi$ for a given set~$W$, we now modify $\varphi$ to a formula $\varphi'$, such that $G^\times \models \varphi'(W)$ if and only if $G^\times_W \models \varphi$. In addition to $V^\times$ and $E^\times$, we also use the free variables $V \coloneqq V(G)$, $D$, and $S$ to identify the vertices of $G$ in $G^\times$, the subdivision vertices of $G^\times$, and the splittable vertices of $G$, respectively.
  
  For every predicate of the form ``$v \in V^\times$'' in $\varphi$, we need to ensure that no vertices of $W$ and no subdivision vertices adjacent to vertices of $V \setminus W$ are allowed. We thus replace every predicate ``$v \in V^\times$'' with the following predicate:
  \begin{equation*}
   v \in V^\times_W \coloneqq v \in V^\times \setminus W \land \lnot(\exists u \in V \setminus W: \textsc{adj}^\times(u, v)).
  \end{equation*}
  We note that the predicate $\textsc{adj}^\times(u, v)$ is true if and only if $u$ and $v$ are adjacent in $G^\times$.
  
  Furthermore, we let the edges of $G^\times$ connecting the adjacent subdivision vertices represent the edges of $G^\times_W$ by replacing the predicate ``$e \in E^\times$'' as follows:
  \begin{equation*}
    e \in E^\times_W  \coloneqq \exists v_1, v_2 \in D: v_1 \neq v_2 \land \textsc{inc}^\times(e, v_1) \land \textsc{inc}^\times(e, v_2).
  \end{equation*}
  The formula $\textsc{inc}^\times(e, v)$ is true if and only if edge $e$ is incident to vertex $v$ in $G^\times$.
  
  Finally, we need to redefine the edge-vertex incidence predicate of $\varphi$ to be consistent with our new edge and vertex predicates from above.
  Since the edges of $G^\times_W$ are represented by edges connecting adjacent subdivision vertices in $G^\times$, we simply need to additionally account for the case where the given vertex is adjacent to one of the endpoints of the specified edge. This corresponds to a vertex of $D$ being contracted into an adjacent vertex of $V \setminus W$ as described earlier.
  \begin{equation*}
    \textsc{inc}^\times_W(e, v)  \coloneqq v \in V^\times_W \land e \in E^\times_W \land (\textsc{inc}^\times(e, v) \lor \exists v' \in D: \textsc{adj}^\times(v, v') \land \textsc{inc}^\times(e, v'))
  \end{equation*}
  
  We remark that the formulas described above can be straightforwardly translated to pure MSO$_2$. Using the following formula on the graph $G^\times$, we can now model whether exploding a set $W \subseteq V(G)$ in the original graph $G$ yields a graph of $\Pi$.
  \begin{equation*}
    \textsc{$\Pi$-Explodable}(W) = W \subseteq S \land \varphi'(W)
  \end{equation*}
  Since, for any fixed MSO$_2$-definable graph class $\Pi$, the corresponding formula $\varphi$ (and thus also~$\varphi'$) has constant size, so does the formula \textsc{$\Pi$-Explodable}.
  We can thus determine whether \textsc{$\Pi$-Explodable} is satisfiable for $G^\times$ in $f(\mathrm{tw}(G^\times)) \cdot n$ time using Courcelle's Theorem~\cite{Courcelle90}.
  Using the optimization version of Courcelle's Theorem~\cite{Arnborg91}, we can determine in the same time whether there exists a set $W$ with $|W| \leq k$ that satisfies this formula. Note that subdividing edges does not change the treewidth of a graph, thus $\mathrm{tw}(G^\times) = \mathrm{tw}(G)$. We therefore obtain the following result.
  \begin{mylemma}
    \label{lm:fptTreewidth}
    For every MSO$_2$-definable graph class $\Pi$, the problem $\explodePi$ is FPT parameterized by the treewidth of the input graph.
  \end{mylemma}
  
  We now again consider the case where the graph class $\Pi$ has bounded treewidth.
  Note that Proposition~\ref{prop:boundedTreewidth} also holds for \explodePi, as the proof can be applied almost verbatim to vertex explosions.
  For any yes-instance $(G, S, k)$ of \explodePi, we thus have $\mathrm{tw(G) \leq \mathrm{tw(\Pi)} + k}$ and we can report any input graph of higher treewidth as a no-instance.
  Since $\mathrm{tw(\Pi)}$ is a constant, we obtain the following result using Lemma~\ref{lm:fptTreewidth}.

  \begin{mytheorem}
    For every MSO$_2$-definable graph class $\Pi$ of bounded treewidth, the problem $\explodePi$ is FPT parameterized by the solution size $k$.
  \end{mytheorem}

  \section{Conclusion}
  In this work, we studied the problems \POVE (\pove) and \POVS (\povs), parameterized by the solution size $k$.
  
  For \pove, we gave an $O(4^k \cdot m)$-time branching algorithm and showed that \pove admits a quadratic kernel 
  that can be computed in linear time. This improves on a recent result by Ahmed et al.~\cite{Ahmed22}, who developed a 
  kernel of size $O(k^6)$ for a more restricted version of the problem.
  
  For \povs, we developed an $O((6k + 12)^k \cdot m)$-time branching algorithm and gave a kernelization algorithm 
  that computes a kernel of size $16k$ in linear time, thus showing that \povs is FPT with respect to the solution size 
  $k$. Interestingly, the branching algorithm for \povs performs significantly worse than its counterpart for \pove, 
  but the kernelization algorithm yields a smaller kernel. This is because, for the \povs problem, the branching 
  algorithm has to additionally consider multiple ways a single vertex can be split. At the same time, however, a 
  single vertex split only eliminates few forbidden substructures, which was a helpful observation to 
  bound the number of vertices in yes-instances for our kernelization.
  
  Finally, we more generally considered the problem of obtaining a graph of a specific graph class $\Pi$ using at most $k$ vertex splits (respectively explosions). For MSO$_2$-definable graph classes $\Pi$ of bounded treewidth, we obtained an FPT algorithm parameterized by the solution size $k$.
  These graph classes include, for example, the outerplanar graphs, the pseudoforests, and the graphs of treewidth (respectively pathwidth) at most~$c$ for some constant $c$.
  
  \bibliography{literature}
    
\end{document}